\documentclass[aip,pop,reprint,amsmath,amssymb,raggedbottom]{revtex4-2}

\usepackage{graphicx}
\usepackage{algorithm}
\usepackage{algpseudocode}
\usepackage{booktabs}
\usepackage{float}
\usepackage{xcolor}
\usepackage[colorlinks=true,linkcolor=blue,citecolor=blue,urlcolor=blue]{hyperref}

\usepackage{amsthm}
\newtheorem{theorem}{Theorem}
\newtheorem{corollary}{Corollary}

\newcommand{\vb}[1]{\mathbf{#1}}
\newcommand{\Jinv}{\vb{J}^{-1}}
\newcommand{\dt}{\Delta t}
\newcommand{\Oh}[1]{\mathcal{O}(#1)}
\newcommand{\Collo}{Collocated Boris}
\newcommand{\Stag}{Staggered Boris}
\makeatletter

\renewcommand{\p@subsection}{}
\def\@subsectioncntformat#1{\csname the#1\endcsname\quad}
\def\@subsubsectioncntformat#1{\csname the#1\endcsname\quad}
\makeatother

\begin{document}

\title{A Collocated Boris Integrator in Flux Coordinates:\\ Balancing Accuracy, Conservation, Cost and Robustness}

\author{Mingyuan Li}
\affiliation{Yuanpei College, Peking University, Beijing 100871, China}
\author{Chang Liu}
\affiliation{Department of Physics, Peking University, Beijing 100871, China}
\date{\today}

\begin{abstract}
When the guiding-center description fails and the full gyromotion must be resolved for energetic particles in complex configurations like stellarators, charged-particle integrators must be formulated directly in the curvilinear flux coordinates. The Boris algorithm, which adopts a staggered scheme in Cartesian coordinates, is phase-space-volume-preserving and second-order accurate; but a direct port to flux coordinates degrades the position update to first order, because the evolving basis vectors of the curvilinear frame make the starting-point metric deviate from the ideal midpoint metric. We construct a collocated, midpoint-predicted Boris algorithm in flux coordinates, restoring second-order accuracy at the cost of one additional field evaluation per step. In reactor-scale stellarator magnetic fields, the scheme recovers second-order convergence in every coordinate component, retains near-machine-precision energy conservation and a bounded magnetic moment, and demonstrates greater orbit robustness than Staggered Boris and RK4 at coarse time steps.
\end{abstract}

\maketitle

\section{Introduction}\label{sec:intro}

Particle-in-cell (PIC) simulation is the workhorse for the multi-scale dynamics of magnetized plasmas, advancing charged particles together with the self-consistent electromagnetic fields~\cite{Birdsall1985}. When the relevant physics involves waves near or above the cyclotron frequency---as in radio-frequency heating and current drive---or energetic particles such as fusion-born alphas, whose finite orbit width and resonant interactions escape the guiding-center (GC) description, the full orbit including gyromotion must be integrated over many cyclotron periods. For this task the Boris algorithm~\cite{Boris1970,Birdsall1985} has remained the de facto standard for over four decades. It is a staggered leapfrog: position and velocity live on interleaved, half-step-offset time layers, and the magnetic field acts on the velocity only through a rotation. Because that rotation preserves the speed, the scheme conserves kinetic energy to machine precision in a static field. Qin \textit{et al.}~\cite{Qin2013} proved that the Boris algorithm, though \textit{not} symplectic, exactly preserves phase-space volume.

Under complex magnetic configurations, however, the Cartesian formulation is rarely convenient. Flux coordinates $\alpha^i=(s,\theta,\phi)$ are far better suited to the three-dimensional geometry of tokamaks and stellarators and to exploiting their field symmetries, so porting the Boris algorithm to such curvilinear systems is highly desirable. Wei \textit{et al.}~\cite{Wei2015} carried out this extension, retaining the staggered scheme in flux coordinates, where it conserves kinetic energy and reproduces guiding-center banana orbits and toroidal precession. Their position advance, however, evaluates the metric at the starting point, $\vb{x}_{n+1}=\vb{x}_n+\dt\,\Jinv(\vb{x}_n)\,\vb{v}_{n+1/2}$, rather than at the trajectory midpoint. As they acknowledge, this introduces an $\Oh{\dt^2}$ error that degrades the coordinate update to first order.

This degradation is geometric in origin. Cartesian coordinates are special because their basis vectors are constant, so the direct midpoint scheme is automatically second-order. In curvilinear coordinates the basis vectors evolve along the trajectory, $\vb{e}_i(\vb{x}_{n+1/2})=\vb{e}_i(\vb{x}_n)+\Oh{\dt}$, and freezing them at the starting point discards this leading $\Oh{\dt}$ variation, leaving a first-order position update even though the Boris velocity rotation remains second-order.

Inspired by Kunz's modified Boris push for the hybrid-kinetic PIC code \textit{Pegasus}~\cite{Kunz2014}---where position and velocity share a time layer and the fields are gathered at a predicted half-step position on a uniform Cartesian grid---we carry this collocated structure into curvilinear flux coordinates. We call the resulting scheme \Collo{}; a single explicit midpoint-prediction step restores the second-order accuracy of the Cartesian Boris method (local truncation error $\Oh{\dt^3}$ in the position) at the cost of one extra field evaluation per step. Across $576$ initial conditions on the quasi-axisymmetric (QA) and quasi-helically symmetric (QH) reactor-scale equilibria of Landreman and Paul~\cite{Landreman2022}, \Collo{} attains a position-error convergence exponent $\alpha\approx3$ in every coordinate component, versus $\alpha\approx2$ for \Stag{} and $\alpha\approx5$ for RK4; it retains the near-machine-precision energy conservation ($\sim10^{-13}$) of Boris-type integrators, keeps the magnetic moment $\mu$ far better bounded than the \Stag{} scheme, and tolerates much coarser steps than RK4. \Collo{} thus combines the orbit-topological correctness of RK4 with the bounded-error conservation of Boris, well suited to PIC simulations running beyond $10^4$ cyclotron periods.

Section~\ref{sec:comparison} compares the three integrators and gives their update formulas; Section~\ref{sec:experiments} reports the numerical experiments; Section~\ref{sec:theory} provides the theoretical error analysis; and Section~\ref{sec:conclusions} concludes. The detailed single-step algorithm and conservation diagnostics are collected in Appendix~\ref{sec:algorithm}.

\section{Structural Comparison of Integration Schemes}\label{sec:comparison}

We benchmark \Collo{} against two integrators: the \Stag{} scheme (Wei \textit{et al.}~\cite{Wei2015}) and RK4. All three advance the same six-dimensional phase-space state
\begin{equation}\label{eq:state}
\vb{x}=(s,\theta,\phi),\qquad \vb{v}=(v_x,v_y,v_z),
\end{equation}
the position $\vb{x}$ in flux coordinates and the velocity $\vb{v}$ in the fixed Cartesian frame. Its equations of motion are
\begin{equation}\label{eq:eom}
\begin{aligned}
\frac{d\vb{x}}{dt} &= \Jinv(\vb{x})\,\vb{v},\\
\frac{d\vb{v}}{dt} &= \frac{q}{m}\,\vb{v}\times\vb{B}_{\mathrm{cart}}(\vb{x}),
\end{aligned}
\end{equation}
where $\vb{J}=\partial\vb{X}_{\mathrm{cart}}/\partial\vb{x}$ is the Jacobian of the map from flux to Cartesian coordinates. Storing the velocity in the Cartesian frame confines all geometric information to the inverse Jacobian in the position update and keeps the velocity rotation free of metric cross-terms.

\subsection{Update Formulas}\label{sec:formulas}

This subsection gives the single-step update of each of the three schemes compared in this work: the \Stag{} scheme, the \Collo{} scheme, and the RK4 method. The two Boris variants share a common velocity rotation $\mathcal{B}$, which advances the velocity through the magnetic field over a step $\dt$,
\begin{equation}\label{eq:boris_rotate}
\begin{gathered}
\vb{t} = \frac{q\dt}{2m}\vb{B}, \qquad \vb{s}=\frac{2\vb{t}}{1+|\vb{t}|^2},\\
\mathcal{B}(\vb{v},\vb{B},\dt) = \vb{v}+(\vb{v}+\vb{v}\times\vb{t})\times\vb{s},
\end{gathered}
\end{equation}
and differ only in \emph{where the inverse Jacobian is evaluated in the position update}; RK4 instead advances position and velocity together.

\medskip

The \Stag{} scheme is a leapfrog on staggered time layers $(\vb{x}_n,\vb{v}_{n+1/2})$; following the leapfrog order it advances the position first and the velocity second:
\begin{align}
\vb{x}_{n+1} &= \vb{x}_n+\dt\,\Jinv(\vb{x}_n)\,\vb{v}_{n+1/2},\\
\vb{v}_{n+3/2} &= \mathcal{B}\bigl(\vb{v}_{n+1/2},\vb{B}(\vb{x}_{n+1}),\dt\bigr).
\end{align}
The position update freezes the metric at the starting point $\vb{x}_n$, half a step behind the midpoint $\vb{x}_{n+1/2}$ of the interval $[\vb{x}_n,\vb{x}_{n+1}]$ it advances; this lag is the $\Oh{\dt}$ mismatch that degrades the position update to first order.

\medskip

The \Collo{} scheme removes this lag with a midpoint-prediction step: it predicts the midpoint $\vb{x}_*$, evaluates the metric there, and advances the position with the symmetric average of the old and new velocities.
\begin{align}
\vb{x}_* &= \vb{x}_n+\tfrac{\dt}{2}\,\Jinv(\vb{x}_n)\,\vb{v}_n, \label{eq:pc_pred}\\
\vb{v}_{n+1} &= \mathcal{B}\bigl(\vb{v}_n,\vb{B}(\vb{x}_*),\dt\bigr), \label{eq:pc_rot}\\
\vb{x}_{n+1} &= \vb{x}_n+\dt\,\Jinv(\vb{x}_*)\,\tfrac{\vb{v}_n+\vb{v}_{n+1}}{2}. \label{eq:pc_corr}
\end{align}
Evaluating the metric at the predicted midpoint $\vb{x}_*$ rather than the starting point captures its leading variation and restores second-order accuracy (Section~\ref{sec:theory}).

\medskip

The RK4 scheme, included as a high-order non-symplectic reference, advances the full state $Z=(\vb{x},\vb{v})$ without distinguishing position from velocity.
\begin{align}
Z_{n+1} &= Z_n+\tfrac{\dt}{6}(\vb{k}_1+2\vb{k}_2+2\vb{k}_3+\vb{k}_4),
\end{align}
with the usual substages $\vb{k}_1=\dot{Z}(Z_n)$, $\vb{k}_2=\dot{Z}(Z_n+\tfrac{\dt}{2}\vb{k}_1)$, $\vb{k}_3=\dot{Z}(Z_n+\tfrac{\dt}{2}\vb{k}_2)$, and $\vb{k}_4=\dot{Z}(Z_n+\dt\,\vb{k}_3)$. It requires four field evaluations per step and is fourth-order accurate, but, lacking any volume-preserving structure, does not bound the energy error over long integrations.

\subsection{Algorithm Summary}\label{sec:algo_summary}

Table~\ref{tab:comparison} collects the structural differences between the three schemes; the accuracy orders they induce are established in Secs.~\ref{sec:experiments} and~\ref{sec:theory}. 

\begin{table}[H]
\centering\small
\caption{Comparison of integration schemes.}\label{tab:comparison}
\begin{tabular}{lccc}
\toprule
Property & \Stag{} & \Collo{} & RK4 \\
\midrule
Time layers & $(x_n,v_{n+1/2})$ & $(x_n,v_n)$ & $(x_n,v_n)$ \\
Metric point & $\vb{x}_n$ (start) & $\vb{x}_*$ (mid) & --- \\
Field evals/step & 1 & 2 & 4 \\
\bottomrule
\end{tabular}
\end{table}

\section{Numerical Experiments}\label{sec:experiments}

All experiments are conducted in flux coordinates through the SIMPLE code~\cite{Albert2020}, which supplies the VMEC/libneo geometric infrastructure ($\vb{J}$, $\Jinv$, $\vb{B}_{\mathrm{cart}}$, $A_\phi$) shared by all three integrators. The equilibria are the QA and QH \texttt{wout} files of Landreman \& Paul~\cite{Landreman2022} (Fig.~\ref{fig:configs}), two contrasting regimes of three-dimensional complexity. 

\begin{figure}[H]
\centering
\includegraphics[width=\columnwidth]{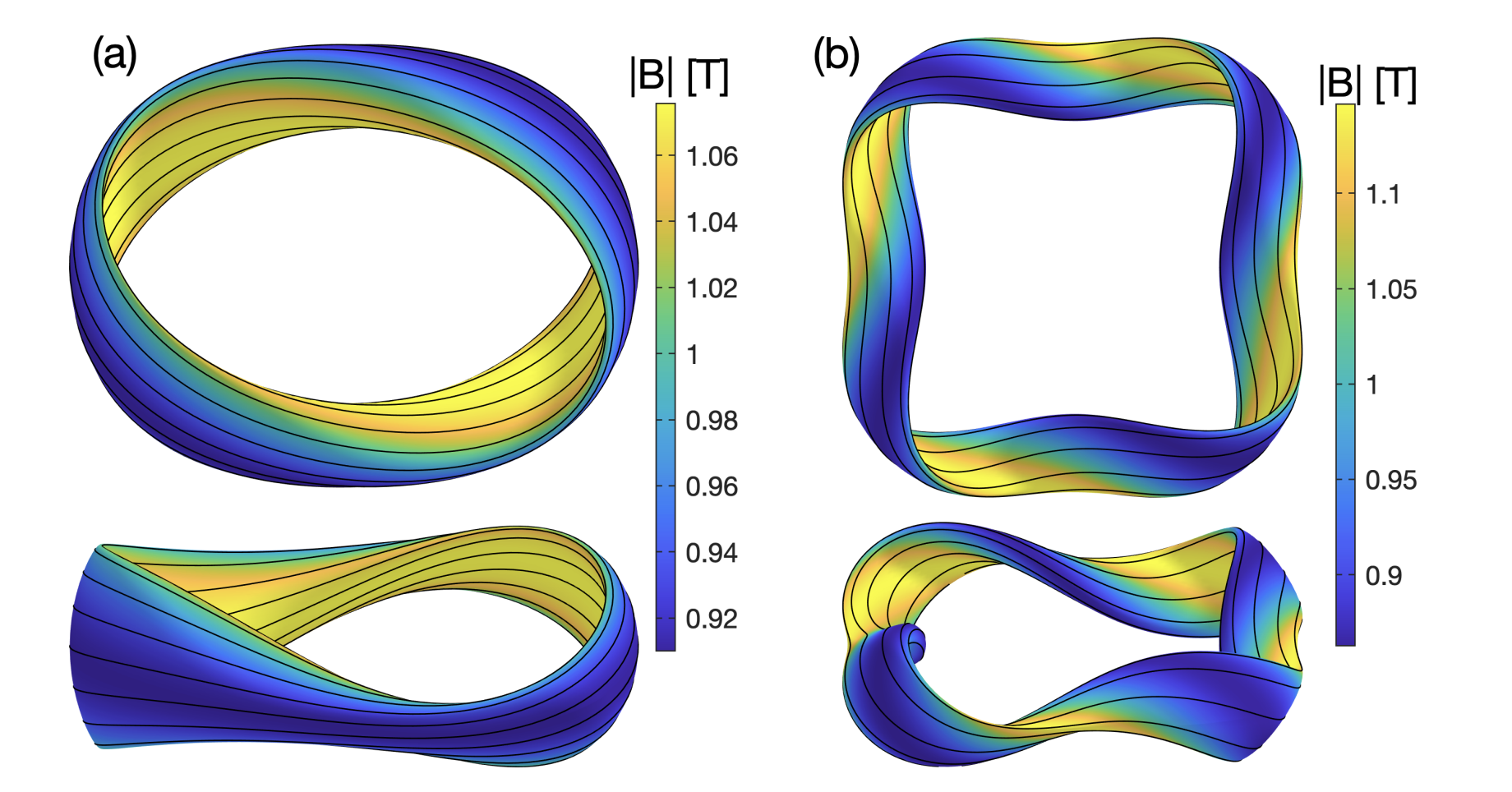}
\caption{Stellarator configurations: QA (left) and QH (right), from Landreman \& Paul~\cite{Landreman2022}.}
\label{fig:configs}
\end{figure}

We begin with the per-component convergence orders, the central quantitative result; we then illustrate their consequences on a benchmark orbit, the conservation diagnostics, and the robustness of each scheme to coarser steps.

\subsection{Error Order Analysis}\label{sec:error_order}

The defining quantitative claim is the per-component convergence order of the three schemes (\Stag{}, \Collo{}, RK4), measured through a single-step protocol. For each step $\dt$ we advance one full step with each scheme from a common initial state and compare against a ``true'' state obtained by integrating to $t=\dt$ with a finely resolved RK4 reference ($\dt_{\mathrm{ref}}=\dt/1000$); the error is the absolute difference per component. Sweeping $\dt$ over a logarithmic ladder of 25 steps spanning $[2\times10^{-7},8\times10^{-2}]$ yields an error-versus-$\dt$ curve whose log--log slope is the exponent $\alpha$, extracted by a sliding-window fit retaining windows with $R^2\geq0.9$. To remove any launch-point dependence, the protocol is repeated over a grid of launch angles $(\theta_0,\phi_0)$ (Table~\ref{tab:params}), drawn as one stratified-random sample per grid cell rather than on the exact grid points. This randomization is essential, not cosmetic: the exact grid points land disproportionately on high-symmetry phases, where the leading truncation coefficient is accidentally small and the fitted exponent is spuriously elevated (an order ``bonus''); sampling one random draw per cell averages can reduce this bias.

The \Stag{} scheme requires special care here: it advances position and velocity on mutually staggered time layers, while both the common initial state $(\vb{x}_0,\vb{v}_0)$ and the reference state $(\vb{x}_1,\vb{v}_1)$ are collocated. We therefore initialize it by advancing one of the two components half a step with the RK4 reference---to $(\vb{x}_{1/2},\vb{v}_0)$ or, equivalently, to $(\vb{x}_0,\vb{v}_{1/2})$---and read from each run the component that lands on the integer time layer: the velocity from the former and the position from the latter. Each measured error thus reflects a single genuine \Stag{} step compared at the same physical instant as the reference, and cannot be inflated above the scheme's true order. The QH exponent map is shown in Fig.~\ref{fig:error_random} and its population statistics in Table~\ref{tab:error_qh_random}; the QA case is given in Appendix~\ref{app:qa_results}.

\begin{figure*}[p]
\centering
\includegraphics[width=\textwidth]{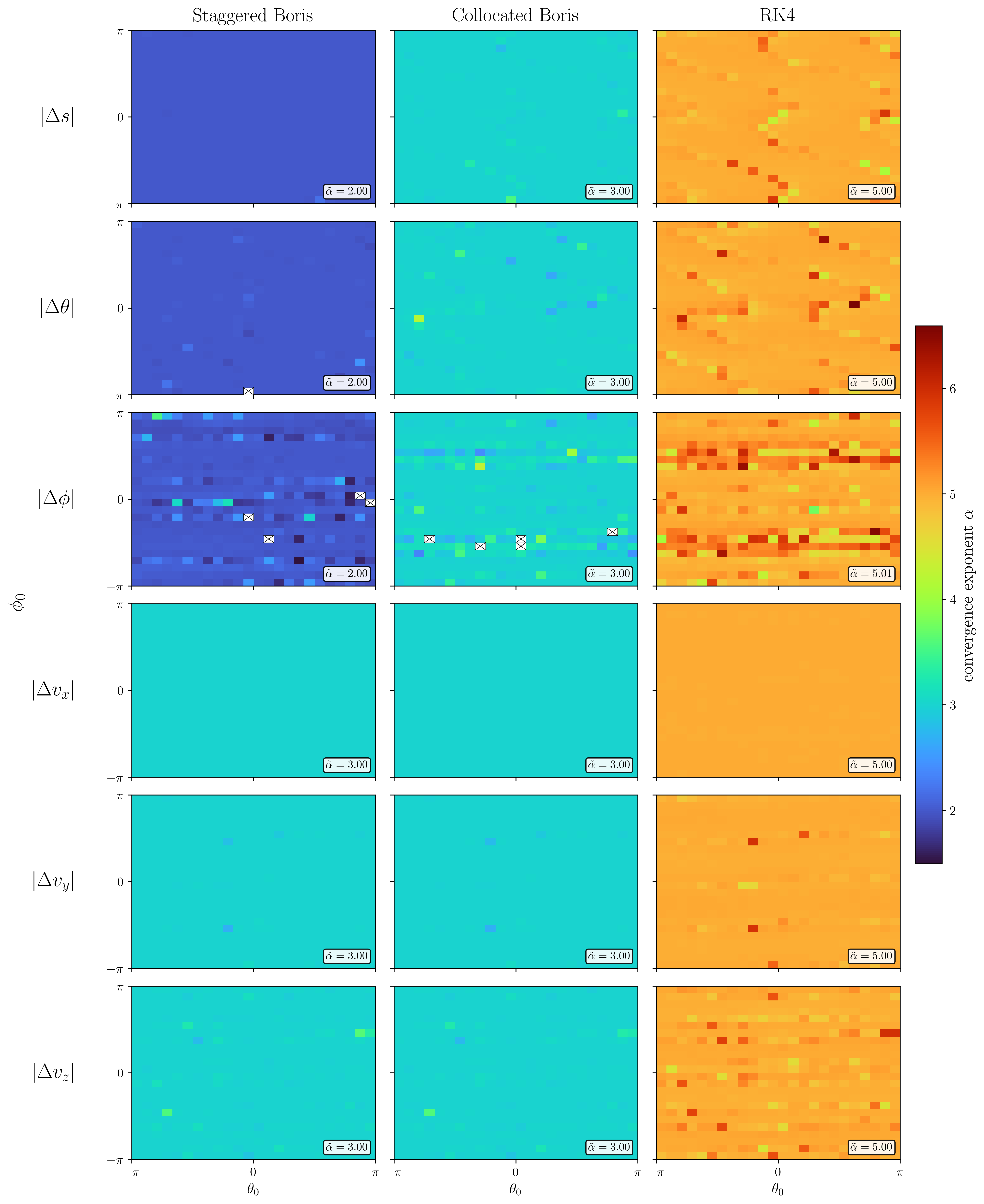}
\caption{Per-component convergence exponent $\alpha$ for QH (stratified-random $24\times24$ sampling in $(\theta_0,\phi_0)$ at $s_0=0.5$). Rows: coordinate and velocity components; columns: \Stag{}, \Collo{}, RK4. Color is the median fitted exponent; each panel is annotated with its median $\tilde\alpha$.}
\label{fig:error_random}
\end{figure*}

\begin{table}[H]
\centering\small
\caption{Convergence exponents $\alpha$ (QH, random $24\times24$, mean$\pm$std, $R^2\geq0.9$, 576 cases).}\label{tab:error_qh_random}
\begin{tabular}{lccc}
\toprule
Component & \Stag{} & \Collo{} & RK4 \\
\midrule
$|\Delta s|$      & $2.00\pm0.01$ & $3.00\pm0.04$ & $5.00\pm0.13$ \\
$|\Delta\theta|$  & $2.00\pm0.03$ & $3.00\pm0.07$ & $5.01\pm0.17$ \\
$|\Delta\phi|$    & $2.02\pm0.16$ & $3.01\pm0.14$ & $5.06\pm0.30$ \\
$|\Delta v_x|$    & $3.00\pm0.00$ & $3.00\pm0.00$ & $5.00\pm0.00$ \\
$|\Delta v_y|$    & $3.00\pm0.02$ & $3.00\pm0.02$ & $5.00\pm0.08$ \\
$|\Delta v_z|$    & $3.00\pm0.04$ & $3.00\pm0.04$ & $5.01\pm0.14$ \\
\bottomrule
\end{tabular}
\end{table}

These measurements admit a clean interpretation. \Collo{} achieves $\alpha\approx3$ across all components, the local-truncation signature of a globally second-order method (Section~\ref{sec:theory}). The \Stag{} scheme yields $\alpha\approx2$ for $s$ and $\theta$, confirming its first-order position update, with a broader $\phi$ distribution reflecting the heightened sensitivity of the toroidal angle to the metric. Crucially, the velocity components of both Boris variants agree at $\alpha\approx3$, demonstrating that the accuracy difference resides entirely in the position pathway. RK4 attains $\alpha\approx5$, consistent with its fourth-order truncation. The maps also show localized order elevation near high-symmetry phases (e.g.\ $\phi_0=0$), where the leading truncation coefficient becomes anomalously small; this pattern strengthens in the more symmetric QA field (Appendix~\ref{app:qa_results}) but averages out of the population means.

\subsection{Benchmark}\label{sec:benchmark}

The benchmark particle is a $1\,\mathrm{MeV}$ proton launched at $s_0=0.5$, $(\theta_0,\phi_0)=(0,0)$ with pitch angle $80^\circ$. The step is given in units of a reference cyclotron period $T_c=2\pi m/(|q|B_{\mathrm{ref}})$ with $B_{\mathrm{ref}}=5.5\,\mathrm{T}$, close to the orbit-averaged $|\vb{B}|$ of this reactor-scale QA equilibrium, so $\dt/T_c$ is the fraction of a gyration per step. At this reference one bounce (banana) period is $\approx5500\,T_c$. We follow the particle with all three schemes (\Stag{}, \Collo{}, RK4) at the common stable step $\dt=1/512$---the coarsest step at which all three still trace an intact banana (see Section~\ref{sec:varstep})---over the full $22000\,T_c$ run (four bounce periods); Fig.~\ref{fig:benchmark} shows the poloidal projection $(x,y)=(\sqrt{s}\cos\theta,\sqrt{s}\sin\theta)$.

\begin{figure*}[t!]
\centering
\includegraphics[width=\textwidth]{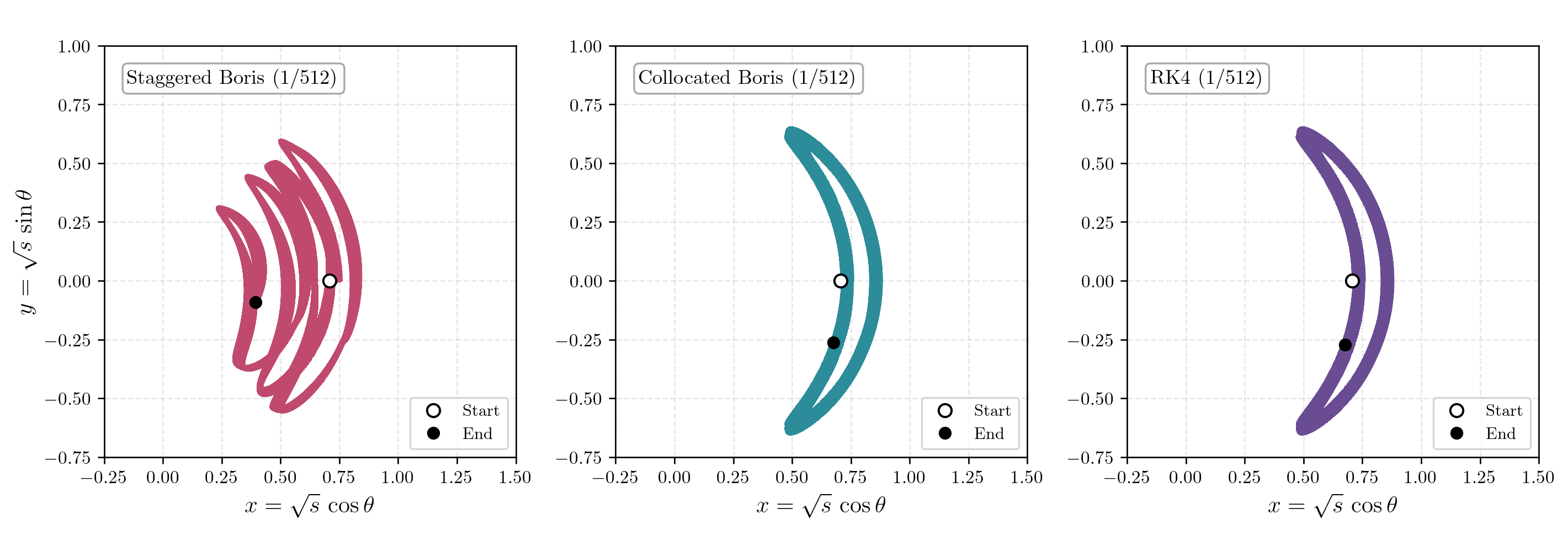}
\caption{QA benchmark orbit at $\dt=1/512$ over $22000\,T_c$ (four bounce periods), poloidal projection $(x,y)=(\sqrt{s}\cos\theta,\sqrt{s}\sin\theta)$ on common fixed axes; open and filled circles mark the launch and final points. Left to right: \Stag{}, \Collo{}, RK4.}
\label{fig:benchmark}
\end{figure*}

\begin{figure}[t!]
\centering
\includegraphics[width=\columnwidth]{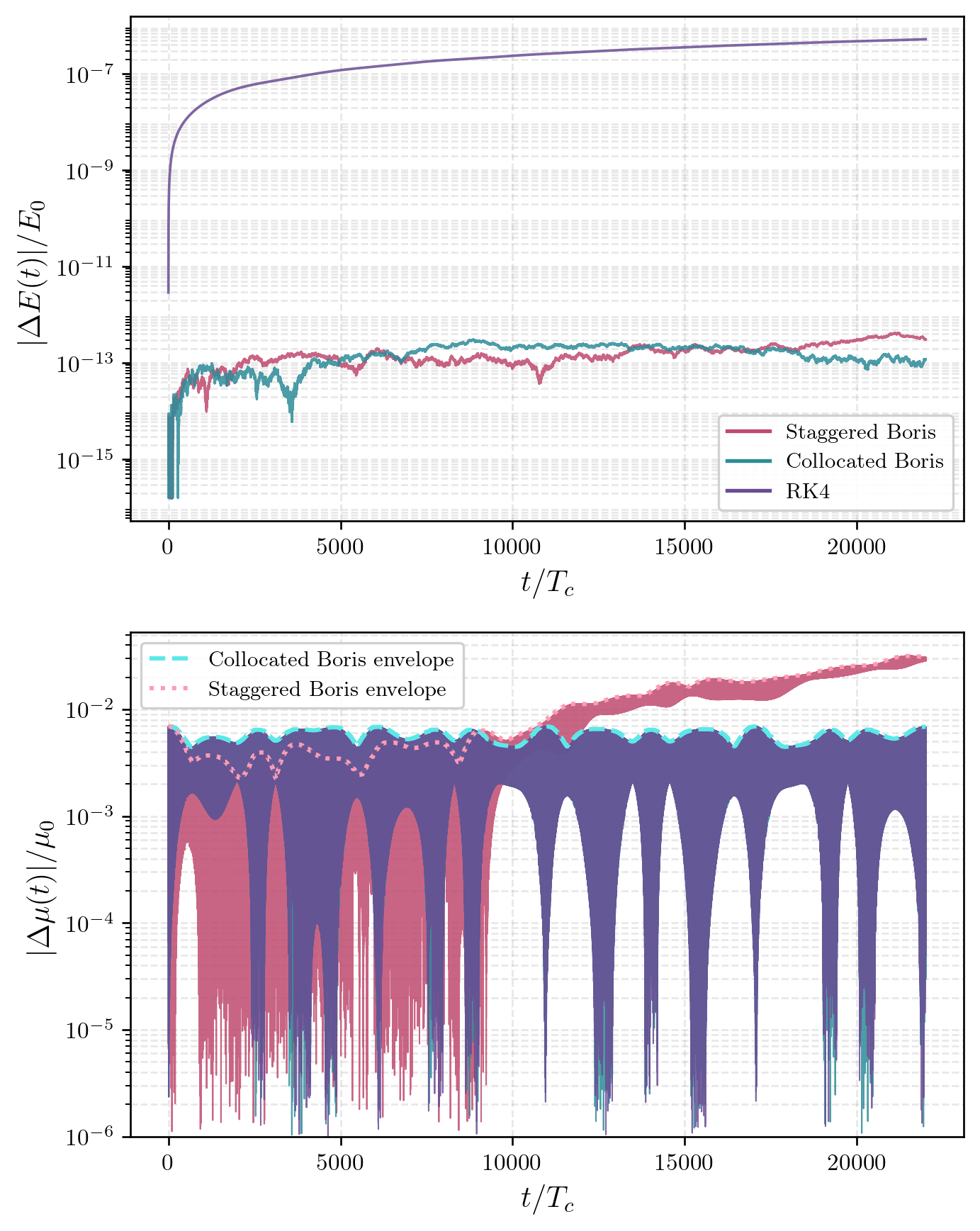}
\caption{Conservation diagnostics for the QA benchmark ($\dt=1/512$, $22000\,T_c$) on logarithmic axes. Top: relative energy error $|\Delta E(t)|/E_0$. Bottom: relative magnetic-moment error $|\Delta\mu(t)|/\mu_0$, with the \Collo{} and \Stag{} upper envelopes overdrawn. Schemes are distinguished by color (see legend).}
\label{fig:diagnostics}
\end{figure}

\Collo{} and RK4 trace the same banana, its bounces overlaid into a slowly drifting envelope, whereas the \Stag{} banana wanders radially inward---the coherent accumulation of its first-order position error. The schemes separate further in their conservation properties (Fig.~\ref{fig:diagnostics}). Both Boris variants hold the relative energy at $\sim10^{-13}$, the machine-roundoff level of a pure rotation, while RK4 drifts to $\sim5\times10^{-7}$. The magnetic moment is more discriminating: \Collo{} and RK4 keep $\mu$ bounded near $0.7\%$, whereas the \Stag{} scheme lets $\mu$ grow monotonically to $\sim3\%$ (Table~\ref{tab:conservation}). The \Stag{} scheme conserves energy but corrupts orbit topology; RK4 preserves the shape but dissipates energy; \Collo{} alone combines orbit correctness with machine-precision energy conservation and a bounded magnetic moment.

\begin{table}[H]
\centering\small
\caption{Conservation summary ($\dt=1/512$, $22000\,T_c$).}\label{tab:conservation}
\begin{tabular}{lccc}
\toprule
Invariant & \Stag{} & \Collo{} & RK4 \\
\midrule
$E/E_0$ & $\sim10^{-13}$ & $\sim10^{-13}$ & drift $5\!\times\!10^{-7}$ \\
$\mu$ & drift $3\!\times\!10^{-2}$ & $\lesssim10^{-2}$ & $\lesssim10^{-2}$ \\
\bottomrule
\end{tabular}
\end{table}

\subsection{Step-Size Robustness}\label{sec:varstep}

Long-time PIC simulations are most economical at the coarsest accurate step. We therefore measure, for each scheme (\Stag{}, \Collo{}, RK4), the coarsest step at which the benchmark particle still traces an intact banana over the full four-bounce horizon. Past this threshold the orbit degrades and no longer traces a faithful banana; coarser still, the integration diverges. The thresholds are listed in Table~\ref{tab:stability}; for each scheme, Fig.~\ref{fig:varstep} shows the orbits computed at step sizes just below and just above its threshold (Appendix~\ref{app:varstep_fig}).

\begin{table}[H]
\centering\small
\caption{Step-size robustness over $22000\,T_c$, $\dt$ referenced to $B_{\mathrm{ref}}=5.5\,\mathrm{T}$.}\label{tab:stability}
\begin{tabular}{lccc}
\toprule
$\dt/T_c$ & \Stag{} & \Collo{} & RK4 \\
\midrule
Intact banana & $<1/512$ & $\le 1$ & $\le 1/32$ \\
Diverges at & $1/256$ & $4$ & $1/2$ \\
\bottomrule
\end{tabular}
\end{table}

The contrast is striking. \Collo{} preserves an intact banana at a step $32\times$ coarser than RK4 ($\dt/T_c=1$ versus $1/32$) and $>500\times$ coarser than the \Stag{} scheme. At $\dt/T_c=1$ the \Collo{} banana is still intact, whereas RK4 holds it only to $1/32$, and the \Stag{} banana drifts even at its finest tested step (Fig.~\ref{fig:varstep}). Beyond its intact-banana threshold \Collo{} remains bounded up to $\dt/T_c=2$ and diverges only at $4$, demonstrating robustness that RK4 (diverging at $1/2$) lacks. The ranking \Collo{} $>$ RK4 $>$ \Stag{} on numerical robustness is therefore established.

Having established where each scheme breaks, we now examine how the surviving orbits depend on step size. Taking the $\dt=1/512$ run as a proxy for the true orbit, we measure for \Collo{} and RK4 the relative RMS deviation of $s$, $E_k$, and $\mu$ over the full $22000\,T_c$ run, at steps below each scheme's divergence threshold (Table~\ref{tab:stability})---up to $\dt/T_c=2$ for \Collo{} and $1/4$ for RK4. Because every step is a power-of-two multiple of $1/512$, all runs share integer-$T_c$ sample times, so the deviation is evaluated point-to-point without interpolation (Fig.~\ref{fig:selfconv}), with the underlying time series in Fig.~\ref{fig:devts}.

\begin{figure}[t!]
\centering
\includegraphics[width=\columnwidth]{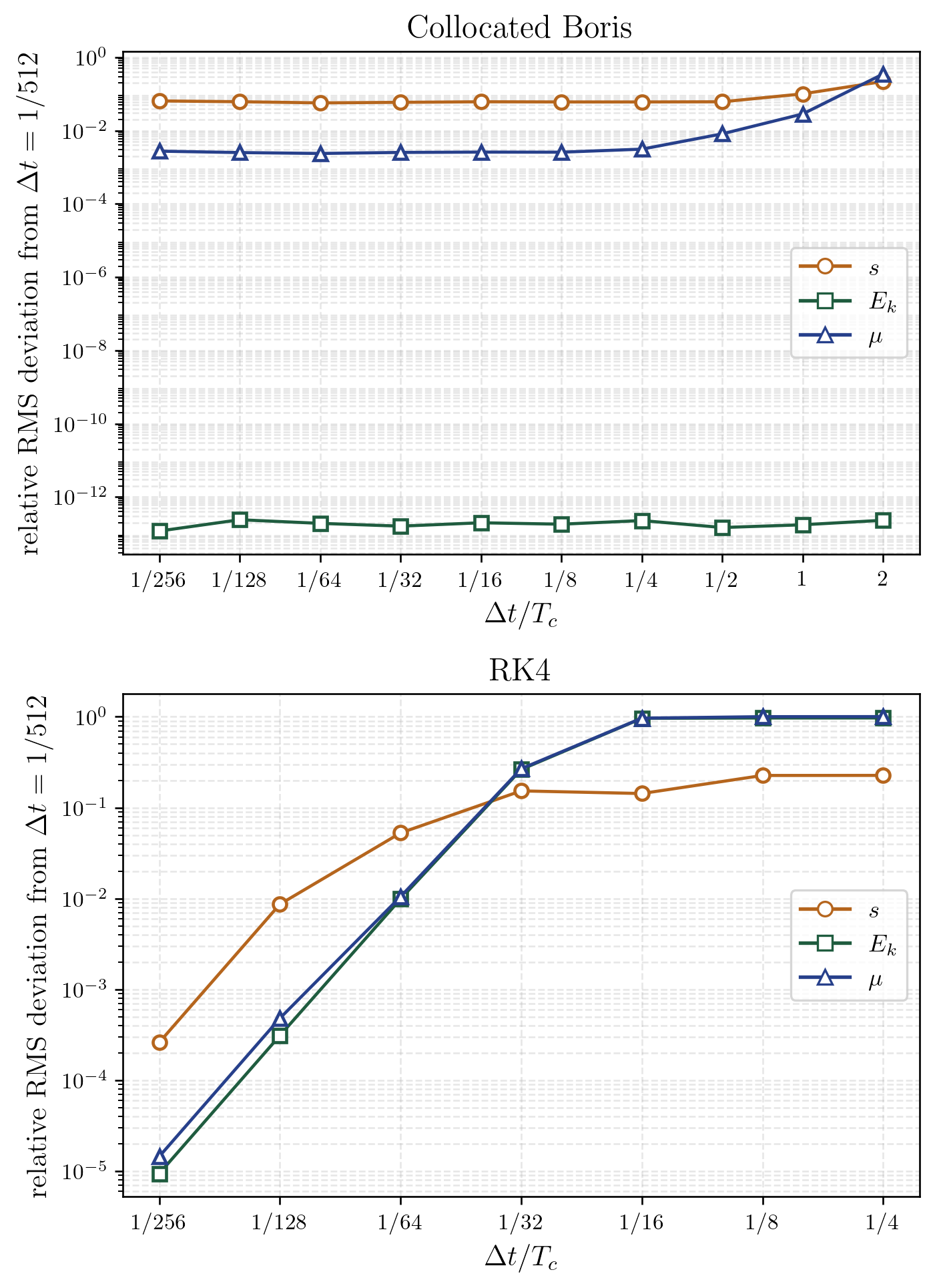}
\caption{Self-convergence to the $\dt=1/512$ reference: relative RMS deviation of $s$, $E_k$, and $\mu$ over the full $22000\,T_c$ run versus step size, for \Collo{} (top) and RK4 (bottom).}
\label{fig:selfconv}
\end{figure}

\begin{figure}[t!]
\centering
\includegraphics[width=\columnwidth]{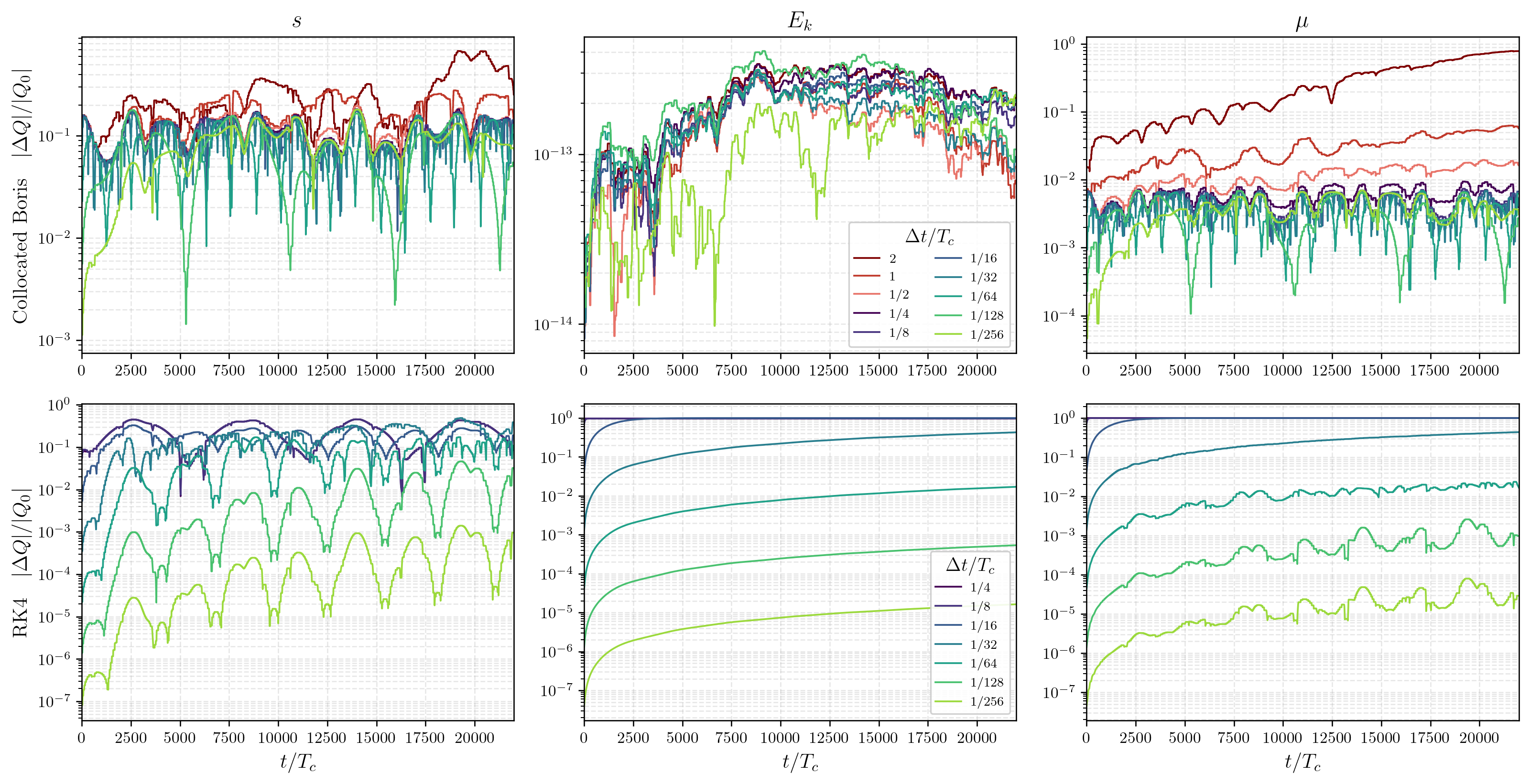}
\caption{Per-step deviation time series from the $\dt=1/512$ reference for \Collo{} (top) and RK4 (bottom), both on log $y$; columns $s$, $E_k$, $\mu$. Each curve is one valid step, color-coded by $\dt$.}
\label{fig:devts}
\end{figure}

The two schemes deviate for different reasons. For \Collo{}, the energy error stays at the $\sim10^{-13}$ machine-precision floor at every step, and the $s$ and $\mu$ deviations sit on a floor ($\sim6\times10^{-2}$ and $\sim2.6\times10^{-3}$) that is nearly independent of step size. This floor is set by the relative gyrophase between the two runs: $s$ and $\mu$ carry a gyro-scale oscillation, and \Collo{}'s second-order phase error saturates the gyrophase randomization at the oscillation amplitude even at the finest steps, so coarsening the step barely changes it. Beyond $\dt/T_c\approx1/2$, however, the deviation departs noticeably from this floor, as the accumulated second-order position truncation error grows and distorts the orbit itself.

RK4 behaves oppositely. At the finest steps (e.g.\ $\dt/T_c=1/256$) its deviation is far smaller than that of \Collo{}, but it rises steeply as the step coarsens (Fig.~\ref{fig:selfconv}), and each single time series grows steadily over the run (Fig.~\ref{fig:devts}). Although its deviation contains the same gyrophase component, this steep rise is dominated by genuine \emph{numerical dissipation}: unlike \Collo{}, which conserves energy to machine precision at every step, RK4 loses energy secularly, and because it is fourth-order per step this dissipation responds strongly to the step size. As the step is refined the curves therefore fall by orders of magnitude and separate clearly across step sizes (Fig.~\ref{fig:selfconv}), converging pointwise toward the reference; at the coarsest bounded steps the dissipation saturates the $E_k$ and $\mu$ errors near order unity.

\section{Theoretical Error Analysis}\label{sec:theory}

The experiments show a single-step error exponent $\alpha\approx3$ for \Collo{} and $\alpha\approx2$ for \Stag{}, i.e.\ second-order versus first-order global accuracy. We now prove that this one-order gap arises entirely from where the metric is evaluated in the position update, and that a predicted midpoint---never the exact one---suffices to close it.

We use the following conventions throughout. Normalized time is written $f(0)\equiv f(t_n)$, $f(1/2)\equiv f(t_n+\dt/2)$, $f(1)\equiv f(t_n+\dt)$. A tilde denotes a numerically computed quantity and an untilded symbol the corresponding exact one. The Boris rotation acts on the Cartesian velocity $\vb{v}=(v_x,v_y,v_z)$; the position $\vb{x}=(s,\theta,\phi)$, however, advances with the \emph{coordinate velocity}
\begin{equation}\label{eq:coordvel}
\vb{u}\equiv\Jinv\vb{v}=\dot{\vb{x}},\qquad u^i=\vb{e}^i\cdot\vb{v},
\end{equation}
whose contravariant components are contractions of $\vb{v}$ with the contravariant basis vectors $\vb{e}^i$ (the rows of $\Jinv$; Appendix~\ref{app:coords}), consistent with the equation of motion~\eqref{eq:eom} and the basis notation of Section~\ref{sec:comparison}. Here $\vb{x}=(s,\theta,\phi)$ carries the coordinates $\alpha^i$ as its components.

As a warm-up, in Cartesian coordinates the exact trajectory is $\vb{x}(1)=\vb{x}(0)+\vb{v}(0)\dt+\tfrac12\dot{\vb{v}}(0)\dt^2+\Oh{\dt^3}$, while the Boris midpoint update is $\vb{x}_B(1)=\vb{x}(0)+\vb{v}(1/2)\dt$ with $\vb{v}(1/2)=\vb{v}(0)+\tfrac12\dot{\vb{v}}(0)\dt+\Oh{\dt^2}$; subtracting gives $\vb{x}(1)-\vb{x}_B(1)=\Oh{\dt^3}$, the familiar second-order accuracy. Constant Cartesian basis vectors make this automatic. In curvilinear coordinates the basis evolves along the trajectory, and the exact coordinate evolution is
\begin{equation}\label{eq:exact_coord}
\begin{aligned}
\alpha^i(1) ={}& \alpha^i(0)+[\vb{e}^i(0)\cdot\vb{v}(0)]\,\dt\\
&+\tfrac12\bigl[\vb{e}^i(0)\cdot\dot{\vb{v}}(0)+\dot{\vb{e}}^i(0)\cdot\vb{v}(0)\bigr]\dt^2+\Oh{\dt^3}.
\end{aligned}
\end{equation}

The argument proceeds in four steps: Theorem~\ref{thm:naive} identifies the first-order defect of the starting-point metric; Corollary~\ref{cor:aligned} shows that an exact midpoint metric would remove it; Corollary~\ref{cor:velpred} shows that the velocity rotation survives replacing the exact midpoint by a predicted one; and Theorem~\ref{thm:pc} shows that the position does too, so the assembled scheme is second-order.

\begin{theorem}[\Stag{} is first-order in position]\label{thm:naive}
Advancing the position with the metric frozen at the starting point gives a local truncation error $\Oh{\dt^2}$ in the coordinates, i.e.\ first-order global accuracy.
\end{theorem}
\begin{proof}
The \Stag{} update evaluates the metric at the start, $\tilde{\alpha}^i(1)=\alpha^i(0)+[\vb{e}^i(0)\cdot\vb{v}(1/2)]\dt$. Expanding $\vb{v}(1/2)=\vb{v}(0)+\tfrac12\dot{\vb{v}}(0)\dt+\Oh{\dt^2}$ and comparing with~\eqref{eq:exact_coord},
\begin{equation}\label{eq:naive_resid}
\tilde{\alpha}^i(1)=\alpha^i(1)-\tfrac12[\vb{v}(0)\cdot\dot{\vb{e}}^i(0)]\,\dt^2+\Oh{\dt^3}.
\end{equation}
The starting-point metric misses the $\tfrac12[\vb{v}\cdot\dot{\vb{e}}^i]\dt^2$ term generated by the evolution of the basis, leaving an $\Oh{\dt^2}$ residual. The position update is therefore first-order, matching the observed $\alpha\approx2$.
\end{proof}

\begin{corollary}[An aligned metric restores second order]\label{cor:aligned}
Evaluating the metric and the velocity at the exact trajectory midpoint reduces the local truncation error to $\Oh{\dt^3}$, i.e.\ second-order global accuracy.
\end{corollary}
\begin{proof}
With the metric and the velocity taken at the exact midpoint, the update reads $\alpha^i(0)+[\vb{e}^i(1/2)\cdot\vb{v}(1/2)]\,\dt$. Expanding both factors to first order, $\vb{e}^i(1/2)=\vb{e}^i(0)+\tfrac12\dot{\vb{e}}^i(0)\dt+\Oh{\dt^2}$ and $\vb{v}(1/2)=\vb{v}(0)+\tfrac12\dot{\vb{v}}(0)\dt+\Oh{\dt^2}$, the cross term supplies exactly the $\tfrac12[\vb{v}(0)\cdot\dot{\vb{e}}^i(0)]\dt^2$ contribution absent from~\eqref{eq:naive_resid}, so that
\begin{equation}\label{eq:star}
\alpha^i(0)+[\vb{e}^i(1/2)\cdot\vb{v}(1/2)]\,\dt=\alpha^i(1)+\Oh{\dt^3}.
\end{equation}
The midpoint metric cancels precisely the residual of Theorem~\ref{thm:naive}.
\end{proof}

\begin{corollary}[The predicted midpoint preserves the velocity order]\label{cor:velpred}
When the midpoint is located by the explicit prediction $\vb{x}_*$, the rotated velocity and its symmetric average satisfy
\begin{equation}
\tilde{\vb{v}}(1)=\vb{v}(1)+\Oh{\dt^3},\qquad
\tilde{\vb{v}}_{\mathrm{avg}}=\vb{v}(1/2)+\Oh{\dt^2}.
\end{equation}
\end{corollary}
\begin{proof}
The predicted midpoint is accurate to second order, $\vb{x}_*=\vb{x}(1/2)+\Oh{\dt^2}$, so the smooth field inherits the same error, $\tilde{\vb{B}}=\vb{B}(\vb{x}_*)=\vb{B}(1/2)+\Oh{\dt^2}$. The Boris rotation vectors carry an explicit factor $\dt$, $\vb{t}=(q\dt/2m)\vb{B}$, so this error is suppressed by one order, $\tilde{\vb{t}}=\vb{t}(1/2)+\Oh{\dt^3}$ and $\tilde{\vb{s}}=\vb{s}(1/2)+\Oh{\dt^3}$. The rotation therefore gives
\begin{equation}
\tilde{\vb{v}}(1)=\vb{v}(0)+\bigl(\vb{v}(0)+\vb{v}(0)\times\tilde{\vb{t}}\bigr)\times\tilde{\vb{s}}=\vb{v}(1)+\Oh{\dt^3},
\end{equation}
and averaging with the initial velocity gives $\tilde{\vb{v}}_{\mathrm{avg}}=\tfrac12[\vb{v}(0)+\tilde{\vb{v}}(1)]=\tfrac12[\vb{v}(0)+\vb{v}(1)]+\Oh{\dt^3}=\vb{v}(1/2)+\Oh{\dt^2}$.
\end{proof}

\begin{theorem}[The \Collo{} scheme is second-order]\label{thm:pc}
With the metric evaluated at the predicted midpoint and the position advanced by the symmetric average velocity, the coordinate update has local truncation error $\Oh{\dt^3}$; together with the second-order velocity rotation, the assembled scheme is second-order.
\end{theorem}
\begin{proof}
The scheme advances the position by $\tilde{\alpha}^i(1)=\alpha^i(0)+[\tilde{\vb{e}}^i(\vb{x}_*)\cdot\tilde{\vb{v}}_{\mathrm{avg}}]\dt$. The midpoint prediction gives $\tilde{\vb{e}}^i(\vb{x}_*)=\vb{e}^i(1/2)+\Oh{\dt^2}$, and Corollary~\ref{cor:velpred} gives $\tilde{\vb{v}}_{\mathrm{avg}}=\vb{v}(1/2)+\Oh{\dt^2}$. Both deviations from the exact midpoint values enter multiplied by $\dt$, hence contribute at $\Oh{\dt^3}$:
\begin{equation}
\begin{aligned}
\tilde{\alpha}^i(1)&=\alpha^i(0)+[\vb{e}^i(1/2)\cdot\vb{v}(1/2)]\,\dt+\Oh{\dt^3}\\
&=\alpha^i(1)+\Oh{\dt^3},
\end{aligned}
\end{equation}
the last equality by Corollary~\ref{cor:aligned}, Eq.~\eqref{eq:star}. The position update is therefore second-order, and by Corollary~\ref{cor:velpred} so is the velocity. The assembled \Collo{} scheme thus carries an $\Oh{\dt^3}$ local truncation error in both position and velocity, i.e.\ second-order global accuracy.
\end{proof}

\section{Conclusions}\label{sec:conclusions}

We have developed the \Collo{} scheme for charged-particle integration in curvilinear flux coordinates, in which the inverse Jacobian is evaluated at a predicted trajectory midpoint---the geometric correction required to keep the position update second-order. This restructuring---replacing the staggered layout with a collocated step whose metric is evaluated at a predicted midpoint rather than the starting point---reduces the position local truncation error from $\Oh{\dt^2}$ to $\Oh{\dt^3}$ at the cost of a second field evaluation per step, one more than \Stag{} but still half of RK4's four. Experiments over the QA and QH equilibria, 576 initial conditions, and 25 step sizes confirm $\alpha\approx3.0$ for \Collo{}, $\approx2.0$ for the \Stag{}, and $\approx5.0$ for RK4, in precise agreement with the theory. \Collo{} matches the orbit-topological correctness of RK4 while retaining machine-precision energy conservation and a bounded magnetic moment, and holds an intact banana at a step $32\times$ coarser than RK4 and $>500\times$ coarser than \Stag{}. A radar-like multi-attribute comparison of the three integrators (Fig.~\ref{fig:radar}) shows that no single scheme dominates every metric, but across accuracy order, per-step cost, energy and moment conservation, and large-step robustness, \Collo{} is the only one without a weak axis. For PIC simulations, \Collo{} thus offers the most favorable trade-off between geometric accuracy, long-term conservation and computation cost.

\begin{figure}[H]
\centering
\includegraphics[width=\columnwidth]{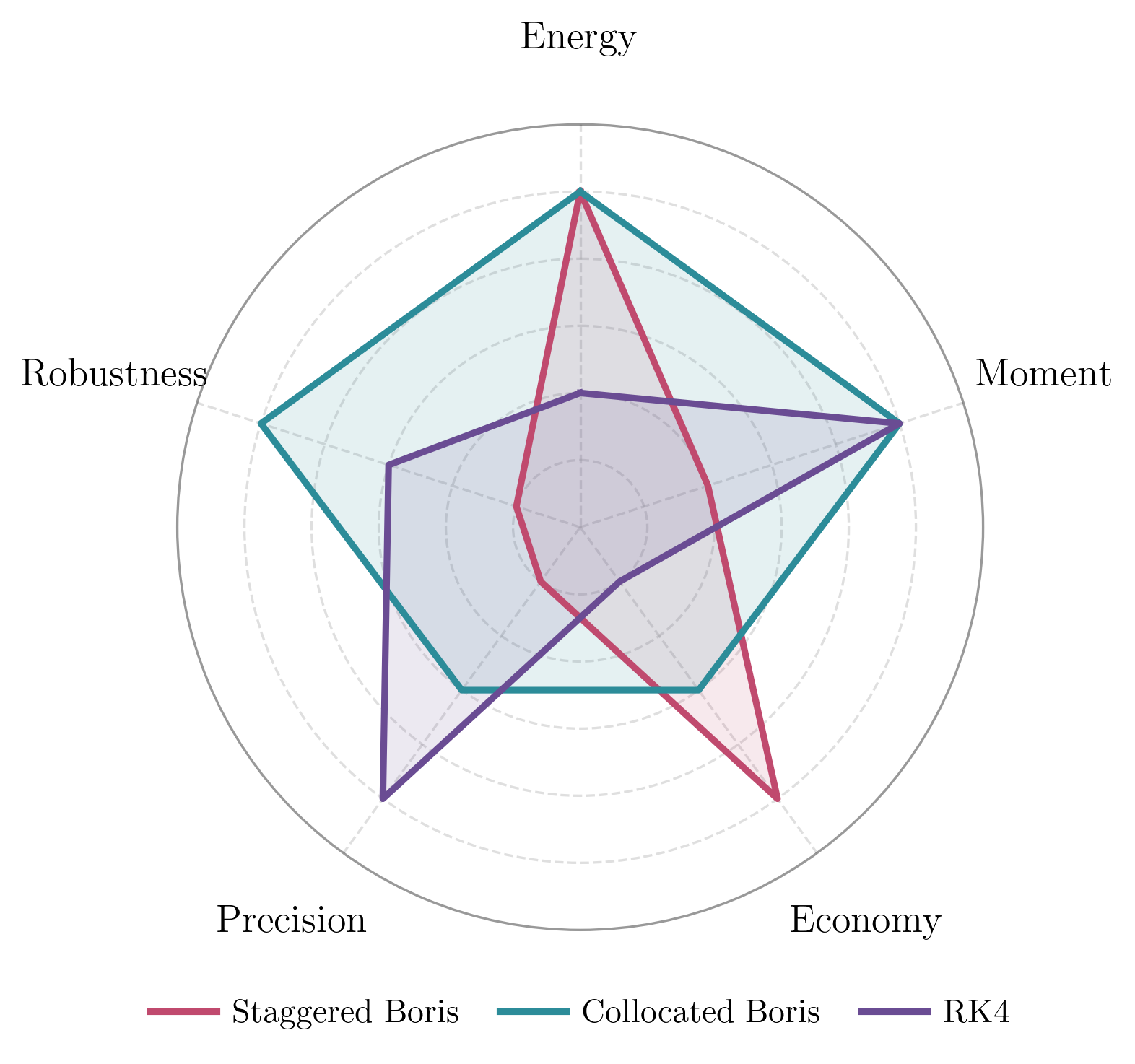}
\caption{Multi-attribute comparison of the three integrators; on every axis the outer edge is the better outcome. \textit{Precision}: accuracy order of the single-step error (Table~\ref{tab:error_qh_random}); \textit{Economy}: field evaluations per step, fewer being better (Table~\ref{tab:comparison}); \textit{Energy} and \textit{Moment}: long-term conservation of energy and magnetic moment (Fig.~\ref{fig:diagnostics}); \textit{Robustness}: retention of an intact orbit at coarse steps, from the survival thresholds and coarse-step self-convergence (Table~\ref{tab:stability}, Fig.~\ref{fig:selfconv}). Scores are an ordinal $1$--$5$ summary of those quantitative results. \Collo{} is the only scheme without a weak axis.}
\label{fig:radar}
\end{figure}

\begin{acknowledgments}
The authors thank the SIMPLE development team for the VMEC field-evaluation infrastructure.
\end{acknowledgments}

\section*{Code Availability}
The minimal code that reproduces the numerical results of this paper---a standalone full-orbit integrator implementing the three schemes compared here, together with the convergence-order experiments---is openly available at \url{https://github.com/Lisenbruth/collocated-boris-flux}.

\appendix

\section{\Collo{} Single-Step Algorithm and Diagnostics}\label{sec:algorithm}

The complete single-step procedure is given in Algorithm~\ref{alg:pcboris}. The scheme evaluates the Boris rotation at the predicted midpoint, and the symmetric velocity average preserves time-reversal symmetry; because the velocity is rotated in the Cartesian frame, the rotation is free of metric cross-terms.

\begin{algorithm}[H]
\caption{The \Collo{} Integrator}\label{alg:pcboris}
\small
\begin{algorithmic}[1]
\Require State $Z_n=(\vb{x}_n,\vb{v}_n)$; step $\dt$; mass $m$; charge $q$
\Statex \textit{--- 1. Midpoint Prediction ---}
\State $\Jinv_n \gets \mathrm{Oracle.InverseJacobian}(\vb{x}_n)$
\State $\vb{u}_n \gets \Jinv_n\vb{v}_n$
\State $\vb{x}_{\mathrm{mid}} \gets \vb{x}_n + \vb{u}_n\,\tfrac{\dt}{2}$
\Statex \textit{--- 2. Boris Cartesian Rotation ---}
\State $\Jinv_{\mathrm{mid}},\vb{B}_{\mathrm{mid}} \gets \mathrm{Oracle.MetricAndField}(\vb{x}_{\mathrm{mid}})$
\State $\vb{t} \gets \tfrac{q\dt}{2m}\vb{B}_{\mathrm{mid}}$
\State $\vb{v}' \gets \vb{v}_n + \vb{v}_n\times\vb{t}$
\State $\vb{s} \gets 2\vb{t}/(1+|\vb{t}|^2)$
\State $\vb{v}_{n+1} \gets \vb{v}_n + \vb{v}'\times\vb{s}$
\Statex \textit{--- 3. Averaged Position Update ---}
\State $\vb{v}_{\mathrm{avg}} \gets \tfrac{1}{2}(\vb{v}_n+\vb{v}_{n+1})$
\State $\vb{u}_{\mathrm{avg}} \gets \Jinv_{\mathrm{mid}}\vb{v}_{\mathrm{avg}}$
\State $\vb{x}_{n+1} \gets \vb{x}_n + \vb{u}_{\mathrm{avg}}\,\dt$
\Statex \textit{--- 4. Diagnostics and State Update ---}
\State $\vb{J}_{n+1},\vb{B}_{n+1},A_{\phi,n+1}^{\mathrm{mag}} \gets \mathrm{Oracle.QueryFull}(\vb{x}_{n+1})$
\State $E_k \gets \tfrac12 m|\vb{v}_{n+1}|^2$
\State $\mu \gets m|\vb{v}_{n+1}\times\hat{\vb{b}}_{n+1}|^2/(2|\vb{B}_{n+1}|)$
\State $P_\phi \gets m(\vb{v}_{n+1}\cdot\vb{J}_{:,3}(\vb{x}_{n+1})) + qA_{\phi,n+1}^{\mathrm{mag}}$
\State \Return $Z_{n+1}=(\vb{x}_{n+1},\vb{v}_{n+1})$ and diagnostics $(E_k,\mu,P_\phi)$
\end{algorithmic}
\end{algorithm}

Each step evaluates three invariants; the accuracy study in the main text uses the first two, with the canonical momentum available but not analyzed there. The kinetic energy $E_k=\tfrac12 m|\vb{v}_{n+1}|^2$ is conserved exactly by any Boris rotation in a static field. The magnetic moment $\mu=m|\vb{v}_{n+1}\times\hat{\vb{b}}_{n+1}|^2/(2|\vb{B}_{n+1}|)$, an adiabatic invariant, is a more stringent test. The canonical toroidal momentum $P_\phi=m(\vb{v}_{n+1}\cdot\vb{J}_{:,3}(\vb{x}_{n+1}))+qA_\phi^{\mathrm{mag}}$ measures respect for toroidal symmetry; strictly conserved only in axisymmetry, its variation in a stellarator is genuine physics that a faithful integrator should keep bounded.

\section{Coordinate System and Metric}\label{app:coords}

In $\alpha^i=(s,\theta,\phi)$ the covariant basis vectors are $\vb{e}_i=\partial\vb{x}/\partial\alpha^i$ with Jacobian $\mathcal{J}=\vb{e}_s\cdot(\vb{e}_\theta\times\vb{e}_\phi)$. Since $\vb{B}\cdot\nabla s=0$ for nested surfaces, the field is tangential, $\vb{B}=B^\theta\vb{e}_\theta+B^\phi\vb{e}_\phi$. From profiles $R(\alpha^i),Z(\alpha^i)$ with $\vb{x}=(R\cos\phi,R\sin\phi,Z)$, the metric is $g_{ij}=\vb{e}_i\cdot\vb{e}_j$:
\begin{equation}
\begin{aligned}
g_{ss}&=R_s^2+Z_s^2, & g_{s\theta}&=R_sR_\theta+Z_sZ_\theta,\\
g_{\theta\theta}&=R_\theta^2+Z_\theta^2, & g_{s\phi}&=R_sR_\phi+Z_sZ_\phi,\\
g_{\phi\phi}&=R_\phi^2+R^2+Z_\phi^2, & g_{\theta\phi}&=R_\theta R_\phi+Z_\theta Z_\phi,
\end{aligned}
\end{equation}
with Jacobian matrix $\vb{J}=[\vb{e}_s,\vb{e}_\theta,\vb{e}_\phi]$, contravariant basis $[\vb{e}^s,\vb{e}^\theta,\vb{e}^\phi]^T=\Jinv$, and $\mathcal{J}=\det\vb{J}=\sqrt{\det g_{ij}}$.

\section{QA Error Order Results}\label{app:qa_results}

The same stratified-random protocol as Section~\ref{sec:error_order}, applied to the QA configuration; the population statistics are collected in Table~\ref{tab:qa_random}.

\begin{table}[H]
\centering\small
\caption{Convergence exponents (QA, random $24\times24$, mean$\pm$std, $R^2\geq0.9$). $^{\dagger}$\Stag{} $|\Delta\phi|$ is elevated by the field symmetry (see text).}\label{tab:qa_random}
\begin{tabular}{lccc}
\toprule
Component & \Stag{} & \Collo{} & RK4 \\
\midrule
$|\Delta s|$      & $2.00\pm0.00$ & $3.01\pm0.09$ & $5.01\pm0.16$ \\
$|\Delta\theta|$  & $2.00\pm0.03$ & $3.00\pm0.05$ & $5.00\pm0.16$ \\
$|\Delta\phi|$    & $2.42\pm0.51^{\dagger}$ & $3.00\pm0.09$ & $5.01\pm0.19$ \\
$|\Delta v_x|$    & $3.00\pm0.06$ & $3.00\pm0.05$ & $5.00\pm0.10$ \\
$|\Delta v_y|$    & $3.00\pm0.01$ & $3.00\pm0.01$ & $5.00\pm0.04$ \\
$|\Delta v_z|$    & $3.02\pm0.14$ & $3.01\pm0.13$ & $5.10\pm0.38$ \\
\bottomrule
\end{tabular}
\end{table}

The QA statistics agree with QH at leading order, confirming that the convergence orders are inherent algorithmic properties. The high QA symmetry strengthens the order-elevation pattern noted in Section~\ref{sec:error_order}: for a growing subset of initial phases the leading truncation coefficient in the toroidal-angle error becomes accidentally small, so the fitted window is partly governed by the next-order term, lifting the marked \Stag{} $|\Delta\phi|$ entry above its neighbors.

\section{Step-Size Robustness Orbits}\label{app:varstep_fig}

For each scheme, Fig.~\ref{fig:varstep} (see next page) shows the poloidal orbit at step sizes just below and just above its intact-banana threshold (Table~\ref{tab:stability}), supporting the discussion of Section~\ref{sec:varstep}.

\begin{figure*}[t!]
\centering
\includegraphics[width=\textwidth]{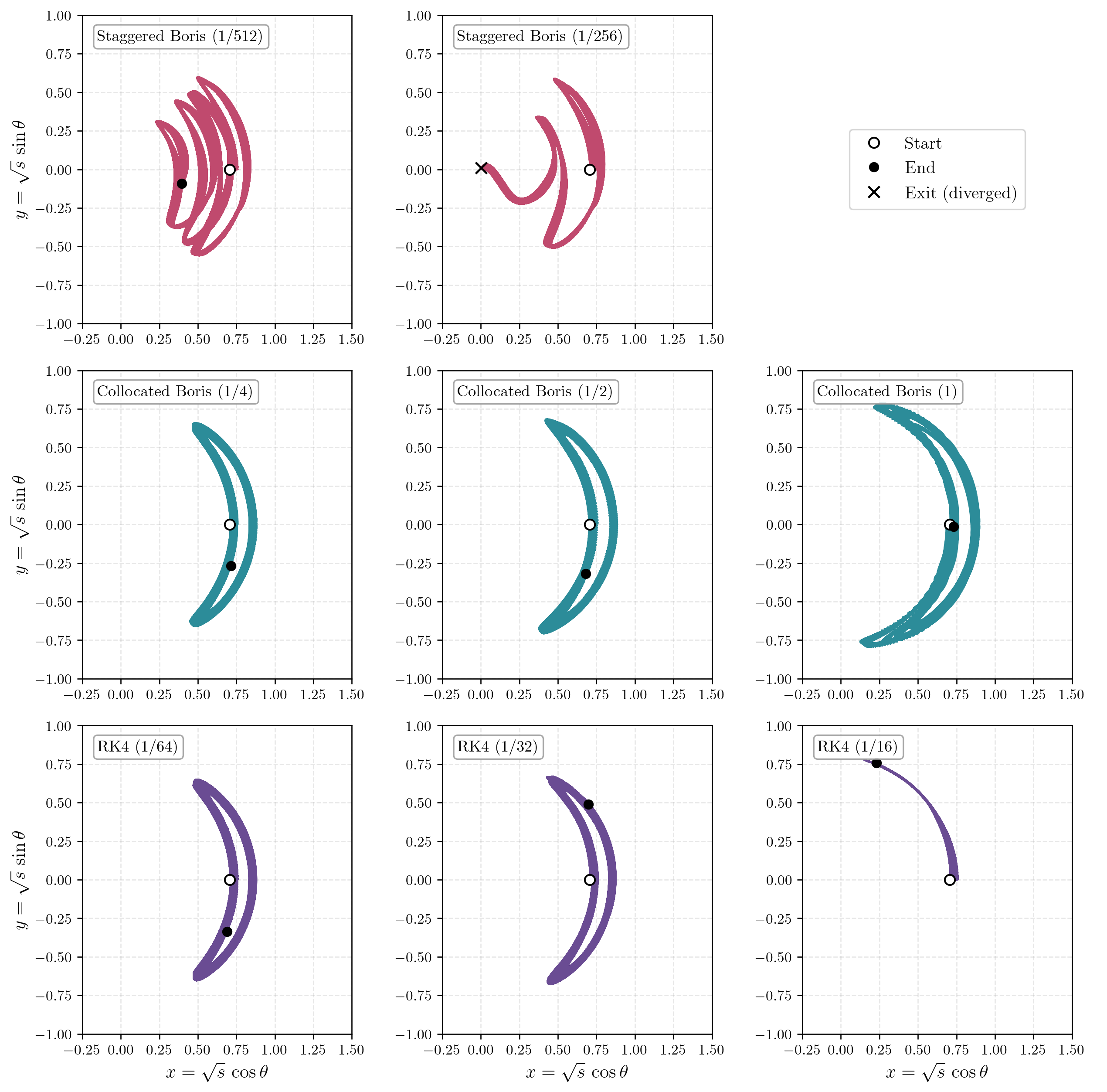}
\caption{Step-size robustness (poloidal projection over $22000\,T_c$, common fixed axes). Rows: \Stag{}, \Collo{}, RK4. Columns are increasing step size: \Stag{} at $\dt/T_c=1/512,1/256$; \Collo{} at $1/4,1/2,1$; RK4 at $1/64,1/32,1/16$. Open and filled circles mark the launch and final points; a cross marks an early exit (divergence). The unused third \Stag{} cell holds the shared legend.}
\label{fig:varstep}
\end{figure*}

\section{Experimental Parameters}\label{app:params}

The QA and QH equilibria are the reactor-scale low-resolution references of Landreman \& Paul~\cite{Landreman2022} (\texttt{wout\_LandremanPaul2021\_}\allowbreak\texttt{\{QA,QH\}\_reactorScale\_}\allowbreak\texttt{lowres\_reference.nc}). Table~\ref{tab:params} lists the run parameters. Time is normalized to $T_c=2\pi m/(|q|B_{\mathrm{ref}})$ with $B_{\mathrm{ref}}$ the tunable input \texttt{ref\_b\_tesla} (default $1\,\mathrm{T}$; here $5.5\,\mathrm{T}$); angles are kept in radians and never wrapped, so large accumulated values ($\sim10^2$ turns) are expected. In Fig.~\ref{fig:configs} the field is shown at the equilibria's native $\sim1\,\mathrm{T}$ normalization; the reactor-scale $B_{\mathrm{ref}}=5.5\,\mathrm{T}$ used in the runs is its orbit-averaged magnitude.

\begin{table}[H]
\centering\small
\caption{Experimental parameters.}\label{tab:params}
\footnotesize
\begin{tabular}{lll}
\toprule
Parameter & Benchmark/step & Error-order \\
\midrule
Configuration & QA & QA, QH \\
$s_0$ & 0.5 & 0.5 \\
Energy / pitch & $1\,\mathrm{MeV}$ / $80^\circ$ & $1\,\mathrm{MeV}$ / $80^\circ$ \\
$(\theta_0,\phi_0)$ & $(0,0)$ & $24\times24$ grid \\
Sampling & --- & stratified random \\
$\dt$ & power-of-two ladder & 25 log $[2\!\times\!10^{-7},8\!\times\!10^{-2}]$ \\
Steps & up to $22000\,T_c$ & 1 (time $=\dt$) \\
Reference & --- & RK4, $\dt/1000$ \\
$B_{\mathrm{ref}}$ & $5.5\,\mathrm{T}$ & --- \\
\bottomrule
\end{tabular}
\end{table}

\bibliography{references}

\end{document}